\documentclass[runningheads]{llncs}
\usepackage[T1]{fontenc}
\usepackage{graphicx}
\usepackage{algorithm}      
\usepackage{algorithmic}
\usepackage{amsmath}
\usepackage{amssymb}
\usepackage[switch]{lineno}
\usepackage{multirow}
\usepackage{booktabs}
\usepackage[colorlinks=true, urlcolor=blue, linkcolor=red]{hyperref}
\begin{document}
\sloppy
\raggedbottom
\title{Boolean Matrix Logic Programming on the GPU}
%
%
\author{
Lun Ai\inst{1,2} 
}
\authorrunning{Ai}
%
\institute{Department of Computing, Imperial College London, London, UK \\
\and Department of Life Sciences, Imperial College London, London, UK\\
\email{lun.ai.public@gmail.com}
}
\maketitle              
\begin{abstract}
Traditional logic programming relies on symbolic computation on the CPU, which can limit performance for large-scale inference tasks. Recent advances in GPU hardware enable high-throughput matrix operations, motivating a shift toward parallel logic inference. Boolean Matrix Logic Programming (BMLP) introduces a novel approach to datalog query evaluation using Boolean matrix algebra, well-suited to GPU acceleration. Building on this paradigm, we present two GPU-accelerated BMLP algorithms for bottom-up inference over linear dyadic recursive datalog programs. We further extend the BMLP theoretical framework to support general linear recursion with binary predicates. Empirical evaluations on reachability queries in large directed graphs and the Freebase 15K dataset show that our methods achieve 1–4 orders of magnitude speed up over state-of-the-art systems. These results demonstrate that Boolean matrix-based reasoning can significantly advance the scalability and efficiency of logic programming on modern hardware. Source code is available on \href{https://github.com/lun-ai/BMLP.git}{https://github.com/lun-ai/BMLP.git}. 

\end{abstract}


\section{Introduction}
\label{sec:intro}

Datalog programs have been an important language for AI systems to integrate structured knowledge  \cite{afrati_linearisability_2003,ceri_datalog_1989,ILP30,evans_learning_2018,gupta_possible_2019,sato_abducing_2018}. Traditional datalog query evaluations primarily focus on symbolic computation. Evaluating first-order logic programs in tensor spaces \cite{cohen_tensorlog_2020,grefenstette_towards_2013,rocktaschel_end_end_2017,sato_embedding_2017,yang_differentiable_2017} has attracted great attention to exploit the fast computation of GPUs. Matrix operations \cite{sato_linear_2017} were shown to significantly improve the efficiency of datalog query evaluations. Despite the compatibility with off-the-shelf frameworks, embedding recursive datalog programs faces several challenges. Datalog programs can be approximated \cite{cohen_tensorlog_2020,grefenstette_towards_2013,rocktaschel_end_end_2017,rocktaschel_injecting_2015,yang_differentiable_2017}, but without recursion or with a limited recursion depth. Rewriting linear datalog programs in linear tensor algebra \cite{sato_linear_2017} also requires parameterised solutions. 

Recent work \cite{muggleton_hypothesizing_2023,ai_boolean_2024} explored Boolean matrix operations in evaluating a subset of datalog programs. These Boolean matrix operations are composable and do not approximate linear recursion. This is shown in Figure \ref{fig:motivation}, where Boolean matrix operations are combined to evaluate a recursive datalog program. In this paper, we expand \textit{Boolean Matrix Logic Programming} (BMLP) \cite{ai_boolean_2024}, a new query answering paradigm in logic programming, where we examine Boolean matrices as a means for logic programming on the GPU.


\begin{figure}[t]
    \centering
        \includegraphics[width=0.7\linewidth]{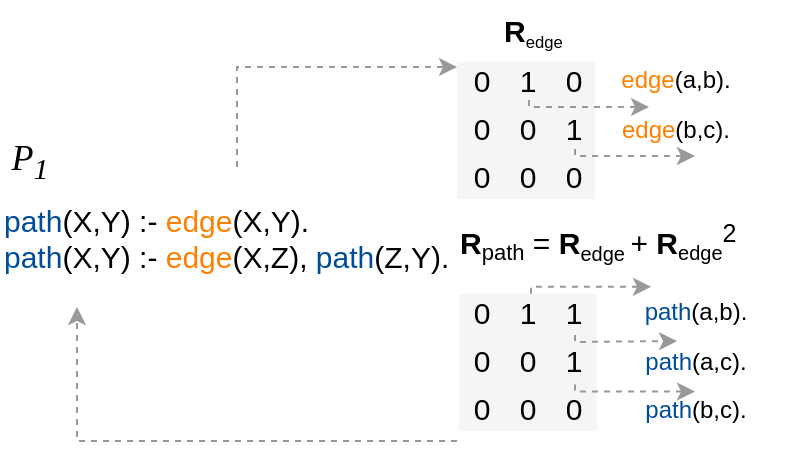} 
    \caption{$\mathcal{P}_1$ describes the transitive closure connectivity between two nodes. $\textbf{R}_{edge}$ and $\textbf{R}_{path}$ are Boolean matrices. }
    \label{fig:motivation}
\end{figure}

We introduce two novel BMLP algorithms for linear recursive programs in the same form as $\mathcal{P}_1$ as basic building blocks: repeated matrix squaring (BMLP-RMS) and selective matrix product (BMLP-SMP). We implemented these two BMLP algorithms in Python with PyTorch for the GPU and SWI-Prolog for the CPU. The runtime of these two BMLP algorithms has been studied theoretically and empirically. We focus on datalog programs where clauses contain predicates of arity at most two ($\mathcal{P}_1$ as an example program), namely the $H_2^2$ program class \cite{muggleton_meta_interpretive_2015}. This program class can be extended with functional semantics to acquire Universal Turning Machine expressivity \cite{tarnlund_horn_1977}. Theoretically, we show that BMLP algorithms can be combined to evaluate arity-two linear datalog programs. We demonstrate the composability of the BMLP approach by combining a BMLP algorithm with Boolean matrix operations and discuss the relationship to second-order programs. Our results present the potential of the BMLP approach to facilitate more efficient inductive \cite{dai_logical_2018,evans_learning_2018,muggleton_meta_interpretive_2015} and abductive learning \cite{dai_abductive_2021,sato_abducing_2018}.
\smallskip


\section{Background}

We refer the readers to \cite{minker_foundations_1988,ILP_foundation} for terminology on logic programs and datalog. Here a definite clause is referred to as a clause. A clause is range-restricted if every variable in the clause head appears among the variables in the clause body. A recursive clause has a predicate in its body that also appears in the head of some clause. A linear recursive clause (e.g. the second clause in $\mathcal{P}_1$) has at most one predicate in its body that appears in the head of some clause \cite{ioannidis_towards_1991}. A linear datalog program (e.g. $\mathcal{P}_1$) only contains linear clauses. The least model is the minimal set containing all facts that are entailed by $\mathcal{P}$. 


$H_n^2$ datalog programs \cite{cropper_logical_2020,muggleton_meta_interpretive_2015} contain predicates of arity at most two and at most $n$ literals in the body of each clause. We consider $H_n^2$ programs that contain binary relations. In relational algebra \cite{codd_relational_1970}, three relational operators are of importance \cite{ioannidis_towards_1991}, which are denoted $S$. In logic programming terms, \textit{cross product} concatenates arguments of one predicate with all arguments of another predicate. \textit{Selection} picks facts of one predicate based on some constraints. \textit{Projection} maps a predicate to a subset of its arguments according to some specifications. 

Finally, matrix algebra implements the transitive closure of binary relations \cite{ceri_datalog_1989}. A set of linear equations can describe the recursive program $\mathcal{P}_1$ in Section \ref{sec:intro}:
\begin{flalign}
    &\textbf{R}_2^{0} = \textbf{0} \qquad \textbf{R}_2^{k} = \textbf{R}_1 + \textbf{R}_1\textbf{R}_2^{k-1}
    \label{eq:closure}
\end{flalign}
where $(\textbf{R}_k)_{ij}$ = 1 if $\mathcal{P} \models r_k(c_i, c_j)$ for constants $c_i, c_j$ and $(\textbf{R}_k)_{ij}$ = 0 otherwise. The transitive closure $\textbf{R}_2^*$ = $\textbf{R}_1 (\textbf{I} + \textbf{R}_2^*)$ has a solution $\textbf{R}_2^*$ = $\sum_{k=1}^{\infty}$\textbf{R}$_1^{k}$ that gives $r_2$ in the least model \cite{aho_universality_1979}. 
Two matrices are equal $\textbf{A} = \textbf{B}$ if $\textbf{A}_{ij} = \textbf{B}_{ij}$ for all $i, j$. Boolean matrix addition ``+'' is defined as (\textbf{A} + \textbf{B})$_{ij}$ = \textbf{A}$_{ij}$ $\lor$ \textbf{B}$_{ij}$. Boolean matrix multiplication ``$\times$'' is defined as $(\textbf{A} \times \textbf{B})_{ij}$ = $\bigvee_{k=0}^n$ \textbf{A}$_{ik}$ $\land$ \textbf{B}$_{kj}$ and is abbreviated as \textbf{A}\textbf{B}. $\textbf{I}$ is the identity and $\textbf{0}$ is the all-zeros matrix. Logical operations ``$\land$'', ``$\lor$'', '``$\neg$''and their derivatives on matrices are performed element-wise, e.g. $(\textbf{A} \land \textbf{B})_{ij} = $ $\textbf{A}_{ij} \land \textbf{B}_{ij}$.


\section{Boolean matrix logic programming}
\label{sec:bmlp_problem}

In contrast to traditional logical inferences that manipulate symbols, BMLP emphasises \textit{programming} Boolean matrix operators. Combining matrix computation facilitates efficient datalog query answering. We present an example of the Boolean matrix encoding and two BMLP algorithm implementations. 
\subsection{Problem definition}
\begin{definition} [BMLP problem]
    Let $\mathcal{P}$ be a $H_n^2$ datalog program containing a set of clauses with predicate symbol $r$. The goal of Boolean Matrix Logic Programming (BMLP) is to output a Boolean matrix $\textbf{R}$ encoded in datalog such that $(\textbf{R})_{ij}$ = 1 if $\mathcal{P} \models r(c_i, c_j)$ for constants $c_i, c_j$ and $(\textbf{R})_{ij}$ = 0 otherwise.
\end{definition}

Boolean Matrix Logic Programming (BMLP) is general for monadic predicates, which can be expressed via dyadic predicates by repeating the single argument. Stored higher-arity predicates are expressible via multiple dyadic predicates. An arbitrary predicate representing input-output pairs of a computation can be mapped to composite constants such as $c_1\_c_2$. A subset of $H_n^2$ programs, namely the $H_2^2$ program class \cite{muggleton_meta_interpretive_2015}, has the same expressiveness as a Universal Turing Machine if it is extended with predicates that represent function input-output pairs \cite{tarnlund_horn_1977}. We can encode the least model as a Boolean matrix mapping constants to a subset of natural numbers. We assume matrices are stored in a database or can be derived directly from non-recursive clauses. 



\subsection{Boolean matrix representation}

A bijective function maps constants in a datalog program $\mathcal{P}$ to a subset of natural numbers $\mathbb{N}_0$. Mapped constants are totally ordered by $\leq$, so each constant is uniquely identified. We encode a Boolean matrix \textbf{R} in datalog to express a clause $r$. Every matrix row $(\textbf{R})_{i,*}$ is a fact $v(i,b_i)$ and $b_i$ is denoted by a binary code such that the j-th bit $(b_i)_j$ is 1 if $\mathcal{P} \models r(c_i, c_j)$ and 0 otherwise. 
\begin{example}
Consider $\mathcal{P}_1$ in Section \ref{sec:intro}. A Boolean matrix $\textbf{R}$ is created from $\{ edge(a,b), edge(b,c)\} \cup \mathcal{P}_1$. Constants $\{a,b,c\}$ are mapped to row and column indices $\{0,1,2\}$. The $v$ facts on the right represent rows in $\textbf{R}$.
\begin{figure}[h]
    \centering
    \vspace{-15pt}
        \includegraphics[width=0.6\linewidth]{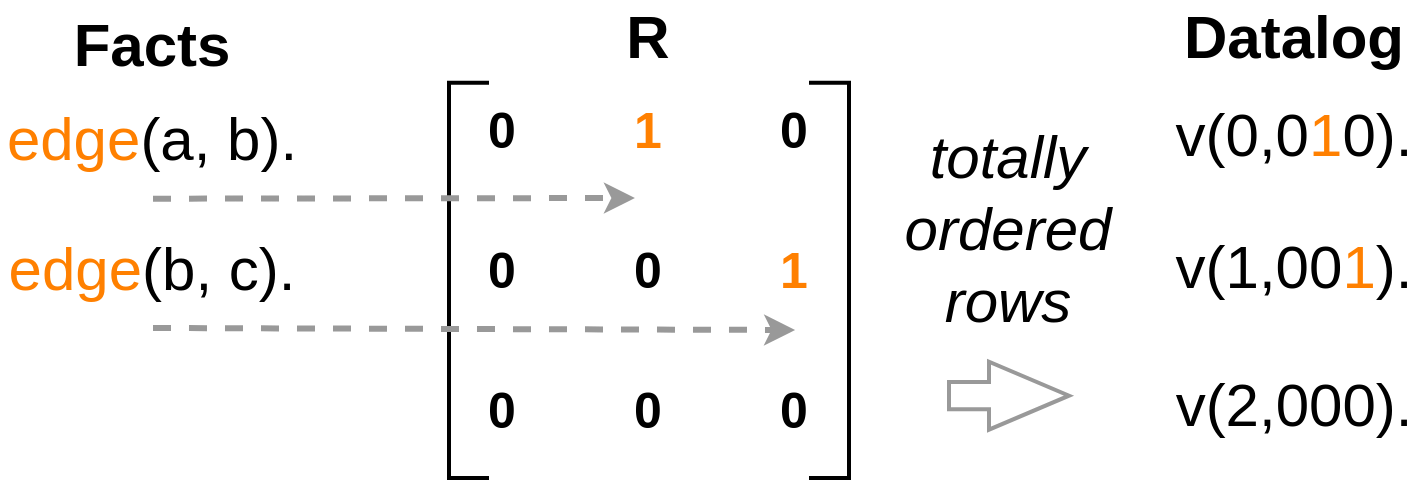} 
    \vspace{-15pt}
\end{figure}
\label{ex:matrix_representation}
\end{example}

When $c_i$ and $c_j$ ($0 \leq i, j < n$) come from the same set of constants, a Boolean matrix created this way is a $n \times n$ square matrix and an operator with the same domain and range. Vectors are treated as single-row matrices. 

\subsection{BMLP algorithms}

\begin{figure}[t]
    \centering
        \includegraphics[width=0.6\linewidth]{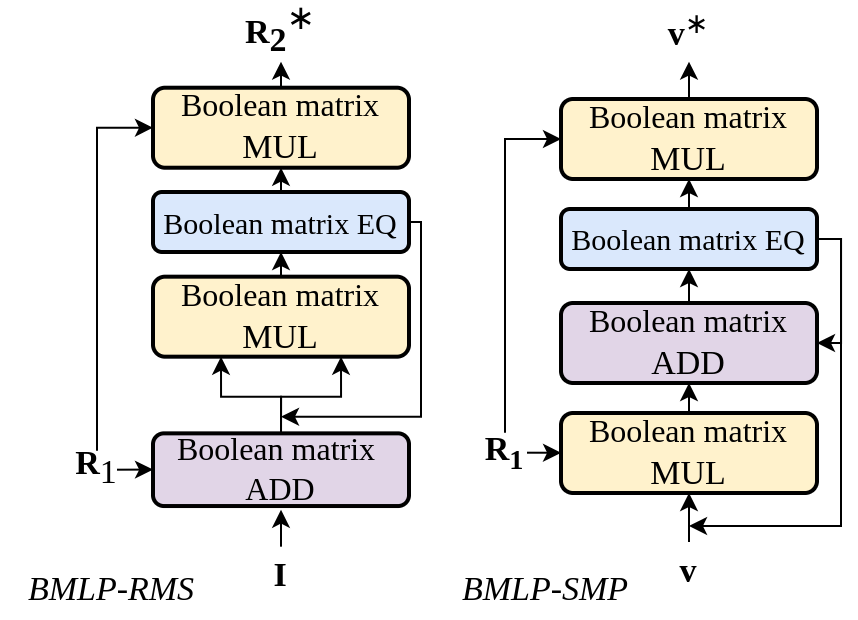} 
    \caption{BMLP algorithms: Repeated Matrix Squaring (BMLP-RMS) and Selective Matrix Product (BMLP-SMP). }
    \label{fig:modules}
\end{figure}

From a high level, each BMLP algorithm aims to map one set of facts to another set via multiple Boolean matrix operations. $\mathcal{C}^k_\mathcal{P}$ denotes the set of all derivable k-arity facts in a datalog program $\mathcal{P}$. This mapping is illustrated by Example \ref{ex:matrix_representation}.
\begin{definition} [BMLP algorithm]
    A BMLP algorithm is an operator $f:2^{\mathcal{C}^2_\mathcal{P}} \to 2^{\mathcal{C}^2_\mathcal{P}}$ where $\mathcal{C}^2_\mathcal{P}$ is encoded by Boolean matrices.
\end{definition}
Let $r_1$ be a stored or directly derivable non-recursive predicate. $r_2$ is a recursive predicate that depends on $r_1$. BMLP-SMP and BMLP-RMS (Figure \ref{fig:modules}) evaluate queries for the following linear recursive $H_2^2$ datalog program $\mathcal{P}_2$:
\begin{gather}
\begin{aligned}
    & r_2(X,Y) \gets r_1(X,Y).  \\
    & r_2(X,Y) \gets r_1(X,Z), r_2(Z,Y). \nonumber
\end{aligned}
\end{gather}
Boolean matrices $\textbf{R}_1$ and $\textbf{R}_2$ encode $r_1$ and $r_2$. BMLP-RMS evaluates all derivable groundings of $r_2$, while the BMLP-SMP algorithm only finds derivable facts from a partially grounded query $r_2(c, Y)$ given some constant $c$.  

\subsubsection{Repeated matrix squaring (BMLP-RMS).} BMLP-RMS computes $\sum_{k=0}^{\infty}$\textbf{R}$_1^{k}$ using the closed semiring of square Boolean matrices with Boolean matrix multiplication, addition, all-zeros matrix \textbf{0} and identity \textbf{I} \cite{droste_semirings_2009}. The transitive closure $\sum_{k=0}^{\infty}\textbf{A}^{k}$ of a squared Boolean matrix \textbf{A} is computable by logarithmic technique in $O(n^3 log_2 n)$ time \cite{fischer_Boolean_1971}: 
\begin{flalign}
    \sum_{k=0}^{\infty}\textbf{A}^{k} = \prod_{k=0}^{\infty} (\textbf{I} + \textbf{A}^{2^k}) = (\textbf{I} + \textbf{A}) (\textbf{I} + \textbf{A}^2) (\textbf{I} + \textbf{A}^4) \, ... 
    \label{eq:repeated_squaring}
\end{flalign} 
Given $n$ constants and facts of the predicate $r_1$, we compile a $n \times n$ matrix $\textbf{R}_1$ such that $(\textbf{R}_1)_{ij}$ = 1 if $\mathcal{P} \models r_1(c_i, c_j)$ and otherwise $(\textbf{R}_1)_{ij}$ = 0. We first consider an alternative form of $\sum_{k=0}^{\infty}$\textbf{R}$_1^{k}$ \cite{ioannidis_computation_1986}:
\begin{flalign}
    \sum_{k=0}^{\infty}\textbf{R}_1^{k} 
    = \lim_{k\to\infty} (\textbf{I} + \textbf{R}_1)^{k} 
    = (\textbf{I} + \textbf{R}_1) (\textbf{I} + \textbf{R}_1) \,...
    \label{eq:simple_recursion}
\end{flalign}

Equation (\ref{eq:simple_recursion}) performs Boolean matrix operations on the same elements in the matrix $(\textbf{I} + \textbf{R}_1)$. To avoid this, we take a similar approach to the logarithmic technique in Equation (\ref{eq:repeated_squaring}) to skip computations by repeatedly squaring matrix products. We re-arrange Equation (\ref{eq:simple_recursion}) as Equation (\ref{eq:repeated_squaring2}):
\begin{flalign}
    \lim_{k\to\infty} (\textbf{I} + \textbf{R}_1)^{k} 
    &=
    (\textbf{I} + \textbf{R}_1) \prod_{k=1}^{\infty}(\textbf{I} + \textbf{R}_1)^{2^k} \nonumber\\
    &=(\textbf{I} + \textbf{R}_1) (\textbf{I} + \textbf{R}_1) (\textbf{I} + \textbf{R}_1)^2 (\textbf{I} + \textbf{R}_1)^4 ... 
    \label{eq:repeated_squaring2}
\end{flalign}
Equation (\ref{eq:repeated_squaring2}) is advantageous over (\ref{eq:repeated_squaring}) and (\ref{eq:simple_recursion}) because it only has to compute the closure $(\textbf{I} + \textbf{R}_1)^{2^k}$ by squaring the matrix product $(\textbf{I} + \textbf{R}_1)^{2^{k-1}}$ from the previous iteration via a single matrix multiplication step. BMLP-RMS used the final result for the transitive closure  $\textbf{R}_2^* = (\sum_{k=0}^{\infty}\textbf{R}_1^{k})\textbf{R}_1$ in Equation (\ref{eq:closure}).

\subsubsection{Selective matrix product (BMLP-SMP).} Compared with BMLP-RMS, the BMLP-SMP algorithm aims to find a specific subset of the least model and does not compute irrelevant facts. For a partially grounded query $r_2(c, Y)$, the constant $c$ is encoded as a $1 \times n$ vector $\textbf{v}$. We consider a $n \times n$ square matrix $\textbf{R}_1$ having the same encoding as for BMLP-RMS. BMLP-SMP selects a subset of elements in $\textbf{R}_2^*$ via $\textbf{v}$:
\begin{flalign}
    \textbf{v}^* = \textbf{v}\textbf{R}_2^* = \textbf{v}\sum_{k=1}^{\infty}\textbf{R}_1^{k} = (\textbf{v} + \textbf{v}\sum_{k=1}^{\infty}\textbf{R}_1^{k}) \textbf{R}_1
    \label{eq:selection}
\end{flalign}
The selection in Equation (\ref{eq:selection}) can be push further into the summation $\textbf{v} + \sum_{k=1}^{\infty} (\textbf{v} \textbf{R}_1^{k-1}) \textbf{R}_1$. This is a Boolean matrix variant of the query-subquery evaluation \cite{vieille_recurswe_1986}. Its computational advantage is that it considers the constant provided in the query as early as possible, and the product of \textbf{v} with $\textbf{R}_1$ at iteration $k$ is calculated from the previous product. BMLP-SMP returns the subset of transitive closure solution $\textbf{v}*$ by multiplying the summation with $\textbf{R}_1$.


BMLP-RMS computes $\textbf{R}_2^*$ that correctly represents $r_2$ facts derivable in $\mathcal{P}_2$, and BMLP-SMP calculates $\textbf{v}^*$ that correctly describes $r_2$ facts derivable in $\mathcal{P}_2$ with the first argument $c$. These follow directly from computing the transitive closure solution $\sum_{k=1}^{\infty}$\textbf{R}$_1^{k}$, so we omit this trivial proof. 

\subsection{BMLP composability}

Multiple BMLP algorithms can be combined by computing, storing and reusing the output of algorithms. Here we focus on BMLP algorithms that compute linear $H_n^2$ datalog programs. 

\begin{definition} [Linear BMLP algorithms]
    A BMLP algorithm is called linear if it computes a linear $H_n^2$ datalog program.
    \label{def:linear_modules}
\end{definition}

BMLP-RMS and BMLP-SMP are linear algorithms by construction. Recall the relational operators in the set $S$, namely cross product, selection and projection. Any operator created from $S$ by concatenations or unions is a linear relational operator in a closed semiring \cite{ioannidis_towards_1991}. Boolean matrix operations in BMLP-RMS and BMLP-SMP have equivalent relation operators, since multiplications and additions are cross products and selections \cite{ioannidis_towards_1991}. 

\begin{theorem}
Every linear $H_n^2$ datalog program can be evaluated by a composition of linear BMLP algorithms. 
\label{theorm:BMLP_linear} 
\end{theorem}
\begin{proof}
Every linear datalog program with range-restricted or non-range-restricted clauses can be computed by concatenations or unions of linear relational operators in $S$ due to the closed semiring   \cite{ioannidis_towards_1991}. A BMLP algorithm corresponding to linear binary relation operators is linear by Definition \ref{def:linear_modules}. 
\end{proof}



We illustrate the composability of BMLP algorithms and Boolean matrix operations via a program that tells if a location is foreign to a region. We consider monadic predicates such as $location$ to indicate the type of entities. Alternatively, they can be represented as dyadic predicates (see Section \ref{sec:bmlp_problem}). Two non-recursive dyadic predicates $contains$ and $adjoins$ are encoded by $\textbf{R}_{contains}$ and $\textbf{R}_{adjoins}$, considering $\{g1,g2,g3,g4,t1,t2,t3\}$ as constants and $contains(t1,g2)$, $contains(g3,t1)$ and $adjoins(g3,g4)$ as facts. The first $H_2^2$ program $\mathcal{P}_3$ represents the transitive closure of the containment relation and is directly computable by the BMLP-RMS algorithm:
\begin{gather}
\begin{aligned}
hasPlace(X,Y) \gets &contains(X,Y). \\
hasPlace(X,Y) \gets &contains(X,Z), hasPlace(Z,Y).\nonumber
\end{aligned}
\end{gather}

Then, the following program $\mathcal{P}_4$ describes whether a location is in a neighbouring area of a region. This program has m-linear recursive clauses. Two clauses invert their arguments and require transposing matrices from the $adjoins$ and $hasPlace$ relations. 
\begin{gather}
\begin{aligned}
indirectlyPartOf(X,Y) \gets &adjoins(X,Y).\\
indirectlyPartOf(X,Y) \gets &adjoins(Y,X).\\
indirectlyPartOf(X,Y) \gets &hasPlace(Z,X), \\
&indirectlyPartOf(Z,Y).\nonumber
\end{aligned}
\end{gather}

The next program $\mathcal{P}_5$ has a single clause $isForeign$. This program contains a stratified negation \cite{dantsin_complexity_2001} and extends the datalog with a new body literal $not \, indirectlyPartOf(X,Y)$. This stratified negated body literal does not involve additional fixpoint computations, but through the difference between the all-ones matrix and the matrix of $indirectlyPartOf$. All groundings of $isForeign$ except $isForeign(t1,g4)$, $isForeign(g2,g4)$, $isForeign(g3,g4)$ and $isForeign(g4,g3)$ are true.
\begin{gather}
\begin{aligned}
isForeign(X,Y) \gets& location(X), location(Y), \\
&not \,\, indirectlyPartOf(X,Y). \nonumber
\end{aligned}
\end{gather}
The combination of linear programs $\mathcal{P}_3 \cup \mathcal{P}_4 \cup \mathcal{P}_5$ is represented by the composition of Boolean matrix operations and BMLP-RMS algorithm in Figure \ref{fig:composition}. The $isForeign$ predicate is encoded by the output matrix $\textbf{R}_{isForeign}$. This contrasts with tensor-based approaches such as \cite{sato_linear_2017} since it is difficult to solve algebraic equations arithmetically that represent m-linear recursive datalog programs. 
\begin{figure}[t]
    \centering
        \includegraphics[width=0.4\linewidth]{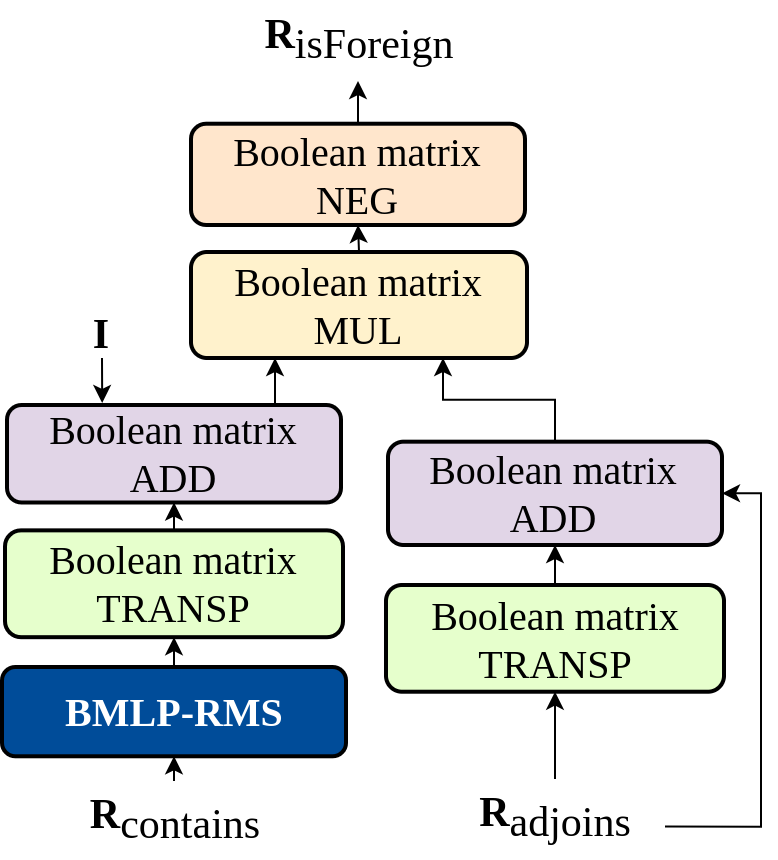} 
    \caption{Combining Boolean matrix operations with BMLP algorithm. TRANSP is the matrix transpose and NEG negates every bit in the input matrix.}
    \label{fig:composition}
\end{figure}

\subsection{BMLP algorithms as second-order programs} A second-order clause contains second-order variables that can be bound to predicate symbols. BMLP algorithms are quantified on matrices compatible in size, where each is bound to a predicate as in a second-order recursive clause. For example, the BMLP-RMS algorithm can represent the following program:
\begin{flalign*}
    & P(X,Y) \gets Q(X,Y).  \\
    & P(X,Y) \gets Q(X,Z), P(Z,Y).
\end{flalign*}
where $\{P,Q\}$ are existentially quantified second-order variables that are instantiated by $hasPlace$ and $contains$ in Figure \ref{fig:composition}. The composition of matrix operations and BMLP-RMS represents additional predicates. Predicate invention \cite{muggleton_machine_1988} aims to automatically derive new and reusable concepts in addition to the background knowledge. Reusing invented predicates can reduce the size of the hypothesis space and improve learning performance. Recent systems invent predicates through high-order language bias \cite{muggleton_meta_interpretive_2015} or constraint generation \cite{ILP30}. BMLP could benefit recent ILP systems to invent predicates by using Boolean matrices to represent two $H^2_2$ second-order clauses. 

\section{Implementation}
\label{sec:implementation}
We implemented BMLP-SMP, BMLP-RMS algorithms and Boolean matrix operators (addition, multiplication, transpose and negation) in Python PyTorch and SWI-Prolog, to accommodate different applications. The SWI-Prolog implementation uses bit-wise Boolean operations and targets single-thread computation on a CPU. On a GPU, the PyTorch implementation uses half-precision floating-point tensor computation from PyTorch with element-wise min-max clamps to emulate boolean matrix operations. Intermediate computation results are reusable to save computation resources using additional storage.


BMLP-RMS computes Equation (\ref{eq:repeated_squaring2}) via Algorithm \ref{alg:algorithm2}. We show that this algorithm has a time complexity of $O(n^3 log_2 n)$ (Appendix \ref{app:time_complexity}). The time complexity of matrix multiplication could be theoretically improved to $O(n^{2.376})$ \cite{coppersmith_matrix_1990}. However, this requires low-level optimisations which may not have a significant return in actual runtime. We do not investigate this tradeoff here. BMLP-SMP calculates Equation (\ref{eq:selection}) by Algorithm \ref{alg:algorithm1}. We show that its time complexity is $O(n^3)$ (Appendix \ref{app:time_complexity}). The vector multiplication derives new facts. The new facts are added to $\textbf{v}$ via Boolean matrix addition until transitive closure is reached.

\begin{algorithm}[t]
    \caption{Repeated matrix squaring (BMLP-RMS)}
    \label{alg:algorithm2}
    \textbf{Input}: A $n \times n$ Boolean matrix $\textbf{R}_1$, identity matrix \textbf{I}.\\
    \textbf{Output}: Transitive closure $\textbf{R}_2^*$.
    \begin{algorithmic}[1]
        \STATE Let $\textbf{R} = \textbf{I} + \textbf{R}_1$, $\textbf{R}_2^*$ = \textbf{0}.
        \WHILE{$True$}
        \STATE $\textbf{R}_2^* = \textbf{R}^2$.
        \STATE $\textbf{R} = \textbf{R}_2^*$ if $\textbf{R}_2^* \ne \textbf{R}$ else break.
        \ENDWHILE
        \STATE $\textbf{R}_2^* = \textbf{R}_2^*\textbf{R}_1$
    \end{algorithmic}
\end{algorithm}

\begin{algorithm}[t]
    \caption{Selective matrix product (BMLP-SMP)}
    \label{alg:algorithm1}
    \textbf{Input}: A $1 \times n$ vector $\textbf{v}$, a $n \times n$ Boolean matrix $\textbf{R}_1$ that encodes facts.\\
    \textbf{Output}: Transitive closure $\textbf{v}^*$.
    \begin{algorithmic}[1]
        \STATE Let $\textbf{v}' = \textbf{v}$.
        \WHILE{$True$}
        \STATE $\textbf{v}^* = \textbf{v}' + \textbf{v}'\textbf{R}_1$.
        \STATE $\textbf{v}' = \textbf{v}^*$ if $\textbf{v}^* \ne \textbf{v}'$ else break.
        \ENDWHILE
        \STATE $\textbf{v}^* = \textbf{v}^*\textbf{R}_1$
    \end{algorithmic}
\end{algorithm}


\section{Experiments}

\subsection{Benchmarks}
We evaluated BMLP algorithms on three benchmarks (Table \ref{tab:datsets}): two reachability tasks in directed graphs with cycles, and a sparse knowledge graph created from Freebase, Freebase 15K. We use a variant of Freebase 15K, FB15K-237 \cite{toutanova_observed_2015}, where inverse relations are not provided and must be evaluated. Directed graphs reachability is computed by the linear recursive $H_2^2$ datalog program $\mathcal{P}_1$. Direct graph datasets were prepared for two tasks: all derivable groundings (DG) and partially grounded queries (DG+partial). In DG, systems computed reachability from the query $path(X,Y)$. In DG+partial, systems evaluated reachability from one target node $path(c1, Y)$. We tested BMLP on sparse matrices on the dataset FB15K-237. We extracted entities as $location$, $contains$ and $adjoins$ relations as facts. We used the program $\mathcal{P}_3 \cup \mathcal{P}_4 \cup \mathcal{P}_5$ to compute derivable $isForeign$ facts.  

\subsection{Compared systems}
As baselines on CPU, we examined Souffle \cite{scholz_fast_2016}, Clingo \cite{gebser_clingo_2014}, B-Prolog \cite{zhou_language_2012} and SWI-Prolog \cite{wielemaker:2011:tplp} and compared with our BMLP SWI-Prolog implementation. Souffle is a datalog engine with state-of-the-art performance even when running on a CPU \cite{sun_optimizing_2025}. Answer Set Programming (ASP) system Clingo is commonly used in logic programming, but its semantics is defined by stable models of ground programs \cite{gelfond_stable_1988}. General-purpose Prolog systems SWI-Prolog and B-Prolog use tabling to ensure termination and faster query evaluation, and B-Prolog has state-of-the-art tabling performance. We tested our PyTorch implementation on Google Colab and used Souffle on Google Colab as a baseline.


\subsection{Preparation}
In DG and DG+partial, we varied the size of programs by $p, p_t \in [0, 1]$ where $p$ is randomly sampled to decide if an edge exists between two nodes $p \leq p_t$. Nodes and edges were mapped to constants and facts in datalog. DG, DG+partial and FB15K-237 datasets were then compiled into Boolean matrices for BMLP algorithms. We recorded each method's mean runtime in evaluating various queries. 

\begin{table}[t]
\caption{The number of nodes ($n$), edges ($k$) and edge sparsity in datasets ($p_t \in [0, 1]$). }
\centering
\begin{tabular}{ c | c | c | c }
 \hline
 Benchmarks & $n$ & $k$ & Sparsity\\ \hline
 DG & 5,000 & $<25,000,000$ & $1 - p_t$\\ 
 DG+partial & 5,000 & $<25,000,000$ & $1 - p_t$ \\
 FB15K-237 & 14,541 & 7,991 & $<0.004\%$\\
 \hline
\end{tabular}

\label{tab:datsets}
\end{table}


\begin{table*}[t]
    \centering
    \caption{Single thread runtimes in seconds. OT: runtime over 15000 seconds. System specification: Ubuntu/Intel(R) Core(TM) i9-7900X CPU @ 3.30GHz with 32GB RAM. }

    \begin{tabular}{l | c c c}
    \hline
     & & DG & \\
    & $p_t=0.01$ & $p_t=0.1$ & $p_t=0.5$ \\ \hline
    BMLP-RMS-SWI      & $281.69 \pm 5.49$ & $\underline{\textbf{248.78}} \pm 4.47$ & $\underline{\textbf{301.63}} \pm 6.39$ \\
    B-Prolog             & $188.23 \pm 1.78$ & $1810.54 \pm 5.13$ & $10148.88 \pm 78.73$ \\
    Clingo                   & $615.56 \pm 20.70$ & $5942.41 \pm 317.45$ & OT \\
    Souffle                 & $\underline{\textbf{61.43}} \pm 0.29$ & $324.91 \pm 0.62$ & $2810.35 \pm 7.79$ \\
    SWI-Prolog           & $843.0 \pm 11.61$ & $8838.07 \pm 41.71$ & OT \\
    \hline
    BMLP speed up & - & \underline{\textbf{0.3 - 35x}} & \underline{\textbf{9 - 49x}}\\
    \hline
    & & DG+partial & \\
    & $p_t=0.01$ & $p_t=0.1$ & $p_t=0.5$ \\ \hline
    BMLP-SMP-SWI     & $\underline{\textbf{0.61}} \pm 0.08$ & $\underline{\textbf{0.71}} \pm 0.09$ & $\underline{\textbf{0.85}} \pm 0.18$ \\
    B-Prolog             & $112.66 \pm 1.71$ & $1130.61 \pm 9.77$ & $6025.05 \pm 57.92$ \\
    Clingo                   & $584.79 \pm 6.5$ & $5976.29 \pm 15.23$ & OT \\
    Souffle                 & $58.09 \pm 1.31$ & $311.55 \pm 1.41$ & $2776.04 \pm 5.04$ \\
    SWI-Prolog           & $663.43 \pm 10.06$ & $7130.39 \pm 35.42$ & OT \\
    \hline
    BMLP speed up & \underline{\textbf{95 - 1087x}} & \underline{\textbf{438 - 10042x}} & \underline{\textbf{3265 - 17647x}}\\
    \hline
    \end{tabular}
    \label{tab:results_cpu}
\end{table*}

\begin{table*}[t]
    \centering
    \caption{Runtimes in seconds. $p_t$ is the probability of an edge between two nodes in a directed graph. System specification: Google Colab (Ubuntu)/Intel(R) Xeon(R) CPU @ 2.20GHz with 84GB RAM/NVIDIA A100 GPU with 40GB memory. }

    \begin{tabular}{l | c c c}
    \hline
     & & DG & \\
    & $p_t=0.01$ & $p_t=0.1$ & $p_t=0.5$ \\ \hline
    BMLP-RMS-PyTorch      & $\underline{\textbf{0.25}} \pm 0.00$ & $\underline{\textbf{0.25}} \pm 0.00$ & $\underline{\textbf{0.26}} \pm 0.00$ \\
    Souffle     & $19.29 \pm 0.19$ & $82.91 \pm 0.47$ & $600.93 \pm 0.50$ \\
    \hline
    BMLP speed up & \underline{\textbf{77x}} & \underline{\textbf{331x}} & \underline{\textbf{2311x}} \\
    \hline
    & & DG+partial & \\
    & $p_t=0.01$ & $p_t=0.1$ & $p_t=0.5$ \\ \hline
    BMLP-SMP-PyTorch     & $\underline{\textbf{0.24}} \pm 0.00$ & $\underline{\textbf{0.24}} \pm 0.00$ & $\underline{\textbf{0.25}} \pm 0.0$ \\
    Souffle     & $16.06 \pm 0.17$ & $86.50 \pm 0.97$ & $620.45 \pm 2.20$ \\
    \hline
    BMLP speed up & \underline{\textbf{66x}} & \underline{\textbf{360x}} & \underline{\textbf{2481x}} \\
    \hline
    \end{tabular}

    \vspace{0.5cm}

    \begin{tabular}{l | c}
    \hline
     & FB15K-237 \\ \hline
    BMLP-RMS-PyTorch       & $\textbf{2.09} \pm 0.04$ \\
    Souffle            & $44.69 \pm 1.38$ \\
    \hline
    BMLP speed up & \underline{\textbf{21x}} \\
    \hline
    \end{tabular}
    \label{tab:results}

\end{table*}

\begin{figure}[t]
    \centering
        \includegraphics[width=0.5\linewidth]{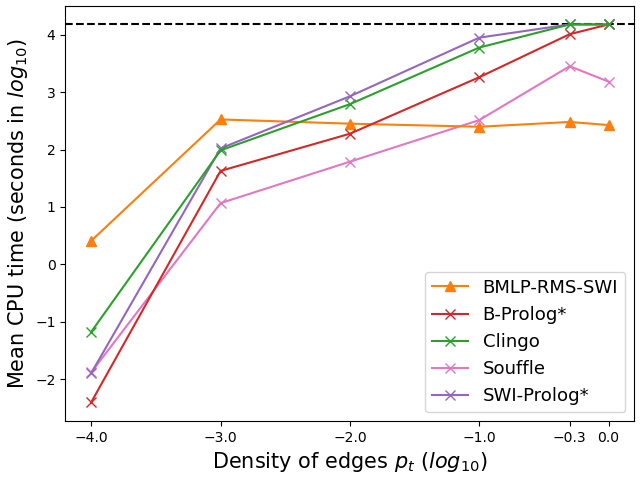} 
    \caption{CPU runtime baselines in DG. We kept $n=5000$ and varied edge density $p_t \in$ \{0.0001, 0.001, 0.01, 0.1, 0.5, 1\}. The grey line is bound at 15000 seconds. }
    \label{fig:exp_1_runtime}
\end{figure}

\begin{figure}[t]
    \centering
        \includegraphics[width=0.5\linewidth]{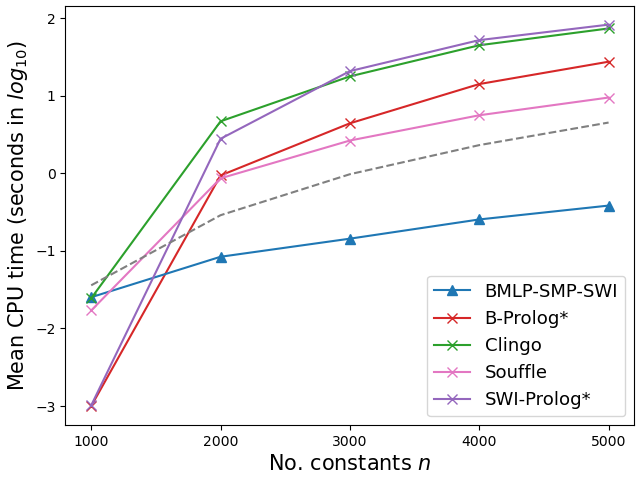} 
    \caption{CPU runtime baselines in DG+partial. We additionally changed $n$ from 1000 up to 5000 at $p_t$ = 0.001. The grey curve grows cubically in $n$.}
    \label{fig:exp_1_runtime_2}
\end{figure}

\subsection{Results} For the CPU baselines, BMLP-RMS and BMLP-SMP speed ups are significant in denser graphs ($p_t >= 0.1$) (Table \ref{tab:results_cpu}). The shortest paths between nodes have lengths closer to 1 in dense graphs, so BMLP-RMS requires fewer matrix operations to compute transitive closure. BMLP-RMS runtime is upper bounded in Figure \ref{fig:exp_1_runtime} with increasing graph density, which supports that BMLP-RMS's runtime is not affected by the number of facts (Proposition \ref{proposition:1}). In Figure \ref{fig:exp_1_runtime_2}, BMLP-SMP's runtime grows with the node number and is bounded by the cubic growth curve, which supports that BMLP-SMP is $O(n^3)$ (Proposition \ref{proposition:2}). 

On Google Colab GPU, BMLP-RMS and BMLP-SMP achieve 20x to 2000x speed up from Souffle, which has the best non-BMLP runtime in the CPU baselines. Owing to half precision computation, BMLP runtime only increases negligibly for increasingly dense datasets. The PyTorch BMLP implementation also enables significant speed ups when datasets are sparse.  


\section{Related work}

\noindent\textbf{Bottom-up datalog evaluation.}  Most systems \cite{afrati_linearisability_2003,beeri_power_1991,ceri_datalog_1989,hu_modular_2022,scholz_fast_2016,tekle_more_2011,vieille_recurswe_1986} used traditional symbol manipulations on the CPU. Obtaining the least model of graph-like datalog programs can be reduced to computing the transitive closure of Boolean matrices \cite{copilowish_matrix_1948,peirce_collected_1932}. Fischer and Meyer \cite{fischer_Boolean_1971} studied a logarithmic divide-and-conquer computation technique by viewing relational databases as graphs and showed a significant computational improvement over direct concatenation of matrix operators. A similar approach was explored by Ioannidis \cite{ioannidis_computation_1986} for computing the fixpoint of recursive Horn clauses. Compared with traditional Inductive Logic Programming (ILP) \cite{ILP1991} which searches for generalisable rules symbolically, Muggleton \cite{muggleton_hypothesizing_2023} employed the logarithmic technique in the DeepLog ILP system. It constructs the bottom clause iteratively by repeatedly squaring Boolean matrices. \smallskip 

\noindent\textbf{Datalog evaluation in tensor space.} Most work explored approximations of non-recursive queries. Rocktaschel et al. \cite{rocktaschel_injecting_2015} combined vector embeddings of entity pairs to evaluate relations with first-order constraints. Several neural-symbolic methods \cite{rocktaschel_end_end_2017,yang_differentiable_2017} computed first-order datalog queries through differentiable unification operators to construct neural networks. In a similar approach, Cohen et al.\cite{cohen_tensorlog_2020} used beliefs propagation for bounded-depth probabilistic inferences in stochastic logic programs, but determining the boundedness is undecidable \cite{gaifman_undecidable_1993}. In contrast, Grefenstette\cite{grefenstette_towards_2013} represented truth values of domain entities and logical operators in a first-order calculus framework, but does not deal with recursions. Sato et al. \cite{sato_embedding_2017,sato_linear_2017,sato_abducing_2018} showed that a subset of linear recursive datalog programs with arity two can be embedded in tensor space as linear algebraic equations. Compared to BMLP, this approach requires parameterised arithmetic computations. 




\section{Conclusion and future work}
We developed two Boolean matrix algorithms capable of evaluating linear recursive datalog on GPUs, building on the Boolean Matrix Logic Programming (BMLP) framework. Theoretically, we established the compositionality of BMLP for linear dyadic recursion, enabling modular reasoning over binary predicates. Empirical evaluations demonstrated that our GPU- and CPU-based implementations outperform state-of-the-art systems by 1 to 4 orders of magnitude. These results highlight the promise of BMLP as a scalable and efficient foundation for logic-based inference, with potential applications in inductive and abductive reasoning tasks. \smallskip

\noindent \textbf{Future work:} Future work may extend the BMLP framework to support multilinear datalog programs, enabling more expressive forms of recursion. Another promising direction involves adapting BMLP algorithms to dynamic databases, where updates and deletions occur over time. In this setting, Boolean matrices could serve as state representations of databases, facilitating efficient incremental computation \cite{hu_modular_2022}. In addition, the efficiency gains offered by BMLP could enhance the scalability of inductive and abductive reasoning, as suggested by recent work on grammar induction using search-based techniques \cite{valizadeh_search_based_2023}.


\section*{Acknowledgments}
The author acknowledges support from the UKRI 21EBTA: EB-AI Consortium for Bioengineered Cells \& Systems (AI-4-EB) award (BB/W013770/1). The author also acknowledges Prof. Stephen H. Muggleton for his valuable input on the initial idea.

\bibliographystyle{splncs04}
\bibliography{ijclr25}

\renewcommand{\thesection}{\Alph{section}}
\appendix
\section*{Appendix}
\section{Notation}
A mode indicator precedes every argument description. For instance, the $sort(+List, -Sorted)$ indicator says this method takes a $List$ and returns it $Sorted$.  An argument mode indicator shows the intended use of the argument and output for a predicate. An argument preceded by ``+'' means that at call time, it must be instantiated to a term. Arguments preceded by ``-'' are output arguments. For additional notations, we refer the readers to SWI-Prolog \cite{swipl}.

\section{BMLP source code and examples}

We refer readers to the GitHub repository (https://github.com/lun-ai/BMLP.git) for the complete implementation of BMLP methods and modules. To showcase, we use the datalog program $ex.pl$ with facts in Example 1 in Section 4.2. The monadic $node$ predicate describes the type of constants:
\begin{verbatim}
node(a). node(b). node(c).
edge(a,b). edge(b,c).
\end{verbatim}

\subsection{BMLP with Python}

BMLP on GPU aims to be Pythonic through PyTorch. This is implemented in the $BMLP\_GPU$ submodule. Several examples can be found in the following section, where many methods are also provided in this Python implementation. 

\begin{verbatim}
import torch
from bmlp.core.tensor import *
from bmlp.core.utils import *

# Extract relations from a Prolog file and create matrices
unary, binary = extract_relations_from_file('bmlp/tests/ex_p0.pl')
data = create_matrices_from_relations('edge',
                                      ['node', 'node'],
                                      unary, binary)

# Convert the data to a PyTorch tensor and apply RMS
m1 = torch.tensor(data['matrix'],
                  dtype=D_TYPE)
m2 = RMS(m1)

# Print the result
print('RMS result:\n', m2)
print_relations(convert_matrix_to_relations(
    m2, data['index_to_entity']['node']))
\end{verbatim}

\subsection{BMLP with SWI-Prolog}
BMLP methods and boolean operations are callable from bmlp.pl as a module in SWI-Prolog.
This module imports source code from the bmlp/ folder to support boolean matrix operations.

\begin{verbatim}
:- use_module(bmlp).

bmlp_ex :- init('./temp'),
           compile('bmlp/tests/ex_p0.pl',db(edge,[node,node],_),M1),
           rms(M1,M2,[output_name='path']),
  
           lm_print(M2).
           
% Calling the goal bmlp_ex prints the output from the BMLP-RMS module.
% path3 (3x3):  
%          a b c
% a       |0 1 1| % path(a, b). path(a, c). 
% b       |0 0 1| % path(b, c).
% c       |0 0 0|
\end{verbatim}

One can convert a matrix into facts by adding the following body. 
This would print out the list [path(a,b), path(a,c), path(b,c)].
\begin{verbatim}
...
lm_to_facts(M2,Fs),
writeln(Fs).
\end{verbatim}

\section{Time complexity}
\label{app:time_complexity}

\begin{proposition} 
    Given a $n \times n$ Boolean matrix $\textbf{R}_{1}$, Algorithm \ref{alg:algorithm2} has a time complexity $O(n^3 log_2 n)$ for computing the transitive closure $\textbf{R}_{2}^*$.
    \label{proposition:1}
\end{proposition}
\begin{proof}
    The time complexity of native multiplications between two $n \times n$ matrices is $O(n^3)$. Finding the closure requires $O(n)$ Boolean matrix multiplications \cite{fischer_Boolean_1971}. Since Algorithm \ref{alg:algorithm2} takes steps to the power of two, this number is reduced to $O(log_2 n)$, giving an overall $O(n^3 log_2 n)$ bitwise operations.
\end{proof}

\begin{proposition}
    Given a $1 \times n$ vector $\textbf{v}$, one $n \times n$ Boolean matrices $\textbf{R}_{1}$, Algorithm \ref{alg:algorithm1} has a time complexity $O(n^3)$ for computing the transitive closure $\textbf{v}^*$.
    \label{proposition:2}
\end{proposition}
\begin{proof}
    The ``while'' loop in Algorithm \ref{alg:algorithm1} runs with $O(n^2)$ bitwise operations due to multiplications between a vector and a $n \times n$ matrix. Until we find the transitive closure, at least one fact containing a new constant needs to be found at each iteration. Therefore, there are at most $n$ iterations which require $O(n^3)$ bitwise operations. 
\end{proof}

\section{Experimentation details}

We record the CPU runtime statistics of all methods in evaluating queries, and repeat each experiment 10 times. B-Prolog and SWI-Prolog have methods, $cputime$ and $call\_time$, respectively for reporting runtimes. Runtime results were extracted from Clingo and Souffle command line and log outputs. Running experiments using the script provided in the GitHub repository generates data for DG, DG+partial and FB15k-237 benchmarks. BMLP modules have been used as shown in the previous section. We show programs evaluated by non-BMLP systems.

\subsection{Directed graphs}

Directed graph datasets, DG and DG+partial, allow us to investigate the runtime variations as a function of the number of facts and the number of constants in the datalog program. These datasets contain dyadic $edge$ and monadic $node$ facts. We only show programs for the benchmark DG since we replaced parts of those programs for the DG+partial benchmark (see commented sections below). Prolog systems do not directly compute the least model, so for B-Prolog and SWI-Prolog, we call the clause $closure$ in the following program to compute derivable facts of $path$ and store them in a table:
\begin{verbatim}
    % enable tabling
    :- table path/2.
    
    path(A,B):-edge(A,B).
    path(A,B):-edge(A,C),path(C,B).

    closure:-path(_C1,_C2), fail.
    % In dataset DB+partial
    % the above line is replaced
    % by the following
    % closure :- path(c1,_C2), fail.
    
    closure.
\end{verbatim}

It is more straightforward with Clingo, which searches for models. We used the answer set program to evaluate and show only $path$:
\begin{verbatim}
    path(A,B) :- edge(A,B).
    path(A,B) :- edge(A,C), path(C,B).
    
    #show path/2.
    % In dataset DB+partial
    % the above line is replaced
    % by the following
    % path_c1(Y) :- path(c1,Y).
    % #show path_c1/1.
\end{verbatim}

Souffle specialises in datalog evaluation and requires definitions of input and output predicates. The input is $edge$ and the output is $path$:
\begin{verbatim}
    .decl edge(n:symbol,m:symbol)
    .input edge
    .decl path(n:symbol,m:symbol)

    path(x,y):-edge(x,y).
    path(x,y):-edge(x,z),path(z,y).
    
    .output path   
    // In dataset DB+partial
    // the above line is replaced
    // by the following
    // .decl closure(n:symbol,m:symbol)
    // closure("c1",x) :- path("c1",x).
    // .output closure
\end{verbatim}

\subsection{Freebase 15K}

We use the Freebase 15K variant (FB15k-237) from \cite{toutanova_observed_2015} because it does not contain the inverse of relations. BMLP architecture would perform with boolean matrix transpose operations to invert relations. In addition, the extensive set of constants but a relatively small amount of facts provides a good stress test for BMLP methods. As mentioned in Section 6.1, we extracted the $contains$ and $adjoins$ facts from FB15k-237 and constants are wrapped by $location$ predicate. 

Non-BMLP systems used the program $\mathcal{P}_3 \cup \mathcal{P}_4 \cup \mathcal{P}_5$ in Section 4.4. We focused on grounding $isForeign$ predicate but derivable facts of involved predicates were also computed by tabling, Clingo or BMLP methods. For SWI-Prolog, we allocated 10GB table space using $set\_prolog\_flag$. We did not have to allocate table space for B-Prolog because it automatically expands it. For both B-Prolog and SWI-Prolog, we used:
\begin{verbatim}
    :- table isForeign/2,
             hasPlace/2,
             indirectlyPartOf/2.
    
    hasPlace(X,Y) :- contains(X,Y).
    hasPlace(X,Y) :- contains(X,Z), 
                     hasPlace(Z,Y).
    indirectlyPartOf(X,Y) :- 
                adjoins(X,Y).
    indirectlyPartOf(X,Y) :- 
                adjoins(Y,X).
    indirectlyPartOf(X,Y) :- 
                hasPlace(Z,X), 
                indirectlyPartOf(Z,Y).
    isForeign(X,Y) :- 
                location(X), 
                location(Y), 
                \+indirectlyPartOf(X,Y).
    
    closure :- isForeign(_C1,_C2), fail.
    closure.
\end{verbatim}

Clingo has to ground the following program which takes a lot of memory. This has been its bottleneck when running on limited memory space compared with BMLP and other non-BMLP methods. 
\begin{verbatim}
    hasPlace(X, Y) :- contains(X, Y).
    hasPlace(X, Y) :- contains(X, Z), 
                      hasPlace(Z, Y).

    indirectlyPartOf(X, Y) :- 
            adjoins(X, Y).
    indirectlyPartOf(X, Y) :- 
            adjoins(Y, X).
    indirectlyPartOf(X, Y) :- 
            hasPlace(Z, X), 
            indirectlyPartOf(Z, Y).

    isForeign(X, Y) :- 
            location(X), 
            location(Y), 
            not indirectlyPartOf(X, Y).
    #show isForeign/2.
\end{verbatim}

For Souffle, the inputs are $contains$, $adjoins$ (capitalised to avoid clashing with Souffle build-in methods) and $location$. The output is $isForeign$. 
\begin{verbatim}
    .decl CONTAINS(x:symbol,y:symbol)
    .input CONTAINS
    .decl hasPlace(x:symbol,y:symbol)
    hasPlace(x,y) :- CONTAINS(x,y).
    hasPlace(x,y) :- CONTAINS(x,z),
                     hasPlace(z,y).

    .decl ADJOINS(x:symbol,y:symbol)
    .input ADJOINS
    .decl indirectlyPartOf(x:symbol,
                           y:symbol)
    indirectlyPartOf(x,y) :- 
            ADJOINS(x,y).
    indirectlyPartOf(x,y) :- 
            ADJOINS(y,x).
    indirectlyPartOf(x,y) :- 
            hasPlace(z,x),
            indirectlyPartOf(z,y).

    .decl location(x:symbol)
    .input location
    .decl isForeign(x:symbol,
                    y:symbol)
    isForeign(x,y) :- 
            location(x),
            location(y),
            !indirectlyPartOf(x,y).

    .output isForeign
\end{verbatim}

\section{Results in Figures}

\begin{table*}[t]
\scriptsize
    \centering
    \caption{Runtime in seconds. Prolog systems have tabling enabled. $p_t$ varies in steps \{0.0001, 0.001, 0.01, 0.1, 0.5, 1\} and $n$ = 5000. OT means a runtime over 15000 seconds. }
    \begin{tabular}{ c @{\hskip 4pt} c @{\hskip 3pt} | c @{\hskip 4.5pt} c @{\hskip 4.5pt} c @{\hskip 4.5pt} c @{\hskip 4.5pt} c }
    \hline
    \multicolumn{2}{c|}{Datasets} & \multicolumn{5}{c}{Methods} \\ 
    Name & \multicolumn{1}{c|}{$p_t$} & BMLP-RMS & B-Prolog & Clingo & Souffle & SWI-Prolog \\ \hline
    & 0.0001 & $2.58 \pm 0.33$ & $\textbf{0.004} \pm 0.001$ & $0.067 \pm 0.016$ & $0.013 \pm 0.001$ & $0.013 \pm 0.005$ \\
    & 0.001 & $333.63 \pm 13.17$ & $42.41 \pm 0.33$ & $96.72 \pm 3.17$ & $\textbf{11.67} \pm 0.31$ & $104.59 \pm 0.88$ \\
    & 0.01 & $281.69 \pm 5.49$ & $188.23 \pm 1.78$ & $615.56 \pm 20.70$ & $\textbf{61.43} \pm 0.29$ & $843.00 \pm 11.61$ \\
    DG & 0.1 & $\textbf{248.78} \pm 4.47$ & $1810.54 \pm 5.13$ & $5942.41 \pm 317.45$ & $324.91 \pm 0.62$ & $8838.07 \pm 41.71$ \\
    & 0.5 & $\textbf{301.63} \pm 6.39$ & $10148.88 \pm 78.73$ & OT & $2810.35 \pm 7.79$ & OT  \\ 
    & 1.0 & $\textbf{266.36} \pm 9.59$ & OT & OT & $1518.34 \pm 4.06$ & OT \\ \hline
    \end{tabular}
    \label{tab:fig4}
\end{table*}

\begin{table*}[t]
\scriptsize
    \centering
    \caption{Runtime in seconds. Prolog systems have tabling enabled. Here, $p_t = 0.001$ and $n$ varies in steps \{1000, 2000, 3000, 4000, 5000\}. OT means a runtime over 15000 seconds. B-Prolog can only print small decimal numbers as 0 for the CPU time statistics.}
    \begin{tabular}{ c @{\hskip 4pt} c @{\hskip 3pt} | c @{\hskip 4.5pt} c @{\hskip 4.5pt} c @{\hskip 4.5pt} c @{\hskip 4.5pt} c }
    \hline
    \multicolumn{2}{c|}{Datasets} & \multicolumn{5}{c}{Methods} \\ 
    Name & \multicolumn{1}{c|}{$n$} & BMLP-SMP & B-Prolog & Clingo & Souffle & SWI-Prolog \\ \hline
    & 1000 & $0.03 \pm 0.01$ & $\textbf{0.0000} \pm 0.0000$ & $0.02 \pm 0.02$ & $0.02 \pm 0.01$ & $0.0005 \pm 0.0003$ \\
    & 2000 & $\textbf{0.08} \pm 0.02$ & $0.94 \pm 0.11$ & $4.67 \pm 0.26$ & $0.86 \pm 0.07$ & $2.78 \pm 1.89$ \\
    DG+partial & 3000 & $\textbf{0.14} \pm 0.02$ & $4.40 \pm 0.28$ & $17.71 \pm 0.55$ & $2.64 \pm 0.18$ & $20.7 \pm 0.98$ \\ 
    & 4000 & $\textbf{0.25} \pm 0.04$ & $14.10 \pm 0.4$ & $44.59 \pm 0.45$ & $5.59 \pm 0.28$ & $51.91 \pm 1.00$ \\ 
    & 5000 & $\textbf{0.38} \pm 0.04$ & $27.36 \pm 0.56$ & $73.41 \pm 0.71$ & $9.45 \pm 0.37$ & $82.21 \pm 1.78$ \\ \hline
    \end{tabular}
    \label{tab:fig5}
\end{table*}

Figure 4 and 5 in the paper visualise runtime changes over the number of facts and constants. Table \ref{tab:fig4} includes runtime data in Figure 4. The number of constants $n$ = 5000 and the probability $p_t$ of an edge between two nodes in a directed graph vary in steps \{0.0001, 0.001, 0.01, 0.1, 0.5, 1\}. Table \ref{tab:fig5} contains runtime results in Figure 5. The probability $p_t$ of facts is 0.001, and the number of constants $n$ varies in steps \{1000, 2000, 3000, 4000, 5000\}. Since B-Prolog can only print small decimal numbers as 0 for the CPU time statistics, we set the runtime lower bound to 0.001 seconds so that values can be plotted with the log scale in Figure 4 and 5. 


\end{document}